\documentclass{IEEEtran}
\usepackage{cite}
\usepackage{amsmath,amssymb,amsfonts}
\usepackage{algorithmic}
\usepackage{graphicx}
\usepackage{textcomp}
\def\BibTeX{{\rm B\kern-.05em{\sc i\kern-.025em b}\kern-.08em
    T\kern-.1667em\lower.7ex\hbox{E}\kern-.125emX}}
    
\usepackage[hyphens]{url}
\usepackage[hidelinks]{hyperref}
\hypersetup{breaklinks=true}
\usepackage{cite}

\usepackage{amsthm,thmtools,color}
\definecolor{nblue}{rgb}{0,0.263,0.576}
\declaretheoremstyle[
  headfont=\color{black}\normalfont\bfseries,
  bodyfont=\normalfont\itshape,
]{colored}
\declaretheorem[
  style=colored,
  name=Proposition,
]{prop}
\declaretheorem[
  style=colored,
  name=Definition,
]{definition}
\declaretheorem[
  style=colored,
  name=Assumption,
]{assumption}

\declaretheorem[
  style=colored,
  name=Lemma,
]{lemma}

\usepackage[compact]{titlesec}

\IEEEaftertitletext{\vspace{-2\baselineskip}}

\begin{document}
\title{Misinformation Regulation in the Presence of Competition between Social Media Platforms (Extended Version)}

\author{So Sasaki, Cédric Langbort
\thanks{Manuscript received January 1, 2023; revised January 2, 2023; accepted January 3, 2023. This work was supported by an ARO MURI award, under agreement W911NF-20-1-0252. }
\thanks{S. Sasaki and C. Langbort are with Coordinated Science Laboratory, 
University of Illinois at Urbana-Champaign, Urbana, IL 61801 USA (e-mail: ssasaki2@illinois.edu, langbort@illinois.edu).}}

\maketitle

\begin{abstract}
Social media platforms have diverse content moderation policies, with many prominent actors hesitant to impose strict regulations. A key reason for this reluctance could be the competitive advantage that comes with lax regulation. A popular platform that starts enforcing content moderation rules may fear that it could lose users to less-regulated alternative platforms. Moreover, if users continue harmful activities on other platforms, regulation ends up being futile. This article examines the competitive aspect of content moderation by considering the motivations of all involved players (platformer, news source, and social media users), identifying the regulation policies sustained in equilibrium, and evaluating the information quality available on each platform. Applied to simple yet relevant social networks such as stochastic block models, our model reveals the conditions for a popular platform to enforce strict regulation without losing users. Effectiveness of regulation depends on the diffusive property of news posts, friend interaction qualities in social media, the sizes and cohesiveness of communities, and how much sympathizers appreciate surprising news from influencers.
\end{abstract}


\section{Introduction}
\vspace{-0.1cm}

In recent years, social media platforms have significantly changed their content moderation policies.
For instance, Jhaver et al. \cite{jhaver2021evaluating} reported that Pinterest recently started actively blocking some search results for controversial or debunked queries, while Reddit has taken action against certain toxic communities by banning or quarantining them. 
Meta/Facebook has also identified and flagged groups that engage in hate speech on its platform. 
Different platforms have implemented widely varying policies, resulting in different standards for what is considered acceptable speech.
One of the most visible examples of this was the different treatment of President Trump's May 2020 messages by Twitter and Facebook, as reported by \cite{washpo} and \cite{wired}.

The difference in content moderation policies between social media platforms is likely influenced by a combination of factors, including each company's interpretation of \S 230 of the US Communications Decency Act and the founders' views on freedom of speech. 
However, it is also probable that economic incentives play a role in shaping these policies. 
In fact, a platform's regulation, or lack thereof, can be perceived as a competitive advantage, as users may choose to switch to a competitor if they feel that the moderation rules in place amount to censorship and limit their access to the news and content they value. 
This effect can be amplified in social networking markets, where the positive externalities of network economics are typically present. 
As more users join a particular platform, it becomes even more socially advantageous for others to join as well, resulting in a winner-takes-all effect, as has been well documented in the literature, such as  \cite{Shapiro1999Information}. Such a threat of seeing its users leave \textit{en masse} may play a role in reducing a platformer's interest in policing misinformation on its network. At the same time, other actors, such as companies advertising on the network, may value regulation (e.g., as a form of so-called 'brand safety' \cite{remi, madio2020user}) and introduce incentives that, in contrast,  may push the platformer towards more oversight.

The aim of this article is to investigate the competitive nature of content moderation on simplified platform models 
by considering the motivations of different players such as the platform owner, news sources, and social media users. We purposefully omit the kind of pro-regulation stakeholders mentioned above in this first study. This is because, in the context of our model, they only make it easier for the platform to implement strict regulation, and we are interested in characterizing the most restrictive such policy that can be sustained in equilibrium.

More precisely, in our model, we examine the effect of regulation when multiple users and a news source (also referred to as an influencer or sender) choose platforms to maximize their individual utility.
A user's utility is determined by their social interaction payoff and news consumption payoff.
Platform owners have the option to deplatform a deceptive news source based on their policies.
As a result, the news source must decide whether to comply with the strict regulation of a popular platform or relocate to an alternative platform to spread more biased information to a smaller audience.

This study is part of a larger body of research on platform governance, as discussed in reviews \cite{Rietveld2020PlatformCA} and \cite{gorwa2019platform}. 
We specifically examine how platforms can prevent the spread of misinformation, such as through the preclusion of accounts. 
Twitter account suspension following the 2020 U.S. presidential election is examined in \cite{chowdhury2021examining}.
Our model also takes into account how users and news sources react to the regulations. 
The way users perceive and interpret content moderation is studied closely in \cite{myers2018censored}. 
Meanwhile, Horta Ribeiro et al. \cite{horta2021platform} investigate how moderation policies affect harmful activity in the wider web ecosystem, and whether this activity becomes more radicalized on alternative platforms. 
Additionally, Innes and Innes \cite{innes2021platforming} examine the intended and unintended consequences of deplatforming interventions, such as the emergence of ``minion accounts'' of a banned account. 

Other recent works concerned with regulation in information diffusion include \cite{tzoumas2012game, bayiz2022countering, {damonte2022targeting}}, all of which place a heavier emphasis on the process underlying dynamic propagation and/or the formation of opinions. 
Little research has been conducted, however, in regard to the intensity of effective regulation that mainstream platforms can implement without losing users.
Understanding this fundamental aspect can advance our conversation on the social responsibilities of tech companies.

This article improves our previous work \cite{sasaki2022necsys} in that 
1) we give tighter sufficient conditions in the propositions for simple social network structures 
and 2) more realistic network structures (finite networks and random networks with community structures) are investigated.
In addition, 3) our analysis is extended to networks consisting of heterogeneous users (sympathizers and non-sympathizers).
These improvements allow our model to analyze e.g., how regulation should be adjusted depending on community sizes, in-community cohesiveness, inter-community structures, and sympathizers near an influencer.

The rest of the paper proceeds as follows.
First, we introduce our model of information diffusion and regulation on a single platform. 
Next, we describe the competitive process of platform migration among the sender and users, and characterize the strictest regulation that can be imposed in equilibrium. 
We then relate its properties to various network structures.
Lastly, the analysis scope is extended to heterogeneous users.

\section{Model}

\subsection{News information from sender}

The game players are multiple users and a sender, who choose platform $\mathbf{P}\in\{\mathbf{A}, \mathbf{B}\}$.
The sender's (respectively user $i$'s) choice of platform $\mathbf{A}$ or $\mathbf{B}$ is denoted as $\mathbf{P}_S$ ($\mathbf{P}_i$).

The sender’s goal is to persuade as many users as possible to adopt a specific belief about the state of the world, while accounting for users' platform choices and updates of beliefs in response to its signaling policy. 
The sender thus acts as a propagandist, whose decision-making process is best described using the framework of Bayesian persuasion \cite{Kamenica2009BayesianP}, as is done in \cite{egorov2020persuasion}.

The world state is represented by a random variable $w\in\{0,1\}$ with $\mu=Prob(w=1)<1/2$.
User $i$ estimates the world state as $a_i\in\{0,1\}$ and gets payoff $c$ if $a_i=w=0$, payoff $1-c$ if $a_i=w=1$, or no payoff if $a_i\neq w$.
The assumption is made that $\mu<c<1/2$, implying that the user's optimal estimation is $a_i=0$ in absence of any news information.
The sender is the only one who can observe the state of the world, and their objective is to have the user estimate $a_i=1$ regardless of the world state.
To persuade users, the sender sends signal $s\in\{0,1\}$.
However, the choice of signal is not entirely arbitrary, and the sender selects a level of deceitfulness, denoted as $\beta\in[0,1]$.
The signal $s$ is generated probabilistically, such that $Prob(s=1|w=0)=\beta, Prob(s=1|w=1)=1$.
This means that the sender reports the true news when the state of the world is $w=1$, but occasionally lies when the world state is unfavorable. 
User $i$ in platform $\mathbf{P}$ receives the signal with probability of $p_{i\mathbf{P}}$, and then makes their estimation $a_i$ based on the sender's strategy $\beta$ and the realized signal $s$.

Now let us examine the optimal strategy for user $i$ with respect to their estimation $a_i$.
If the user does not receive any signal, they should choose the default optimal estimation of $a_i=0$.
If they receive signal $s=0$, the biased sender is reporting unfavorable news and the world state is without a doubt $w=0$.
However, if they receive signal $s=1$, which could potentially be misleading, a more careful analysis is needed.
In this case, the expected payoff for the user when choosing $a_i=1$ is given by 
\begin{equation}
    Prob(w=1|s=1)(1-c) = \frac{\mu}{\mu+(1-\mu)\beta}(1-c),
\end{equation}
while for $a_i=0$, it is given by
\vspace{-0.2cm}
\begin{equation}
    Prob(w=0|s=1)c = \frac{(1-\mu)\beta}{\mu+(1-\mu)\beta}c.
\end{equation}
Therefore, the user should trust signal $s=1$ and select $a_i=1$ if the sender's deceitfulness satisfies
$\displaystyle{    \beta \le \frac{\mu(1-c)}{(1-\mu)c} =: \beta'. }$
Otherwise, they should distrust signal $s=1$ and select $a_i=0$.
Overall, the user's optimal strategy when receiving a signal is to trust any signal ($a_i=s$) if $\beta\le\beta'$, and to ignore any signal ($a_i=0$) otherwise.
Therefore, their expected payoff (for estimating the world state correctly) in platform $\mathbf{P}$ is given by 
\begin{equation}
    \Psi_{i\mathbf{P}} = 
    \begin{cases}
        \mu p_{i\mathbf{P}} (1-c) + (1-\mu)(1-\beta p_{i\mathbf{P}})c & \text{if} \beta\le\beta'\\
        (1-\mu)c & \text{if} \beta>\beta'
    \end{cases}
\end{equation}
We call $\Psi_{i\mathbf{P}}$ news consumption payoff.
The sender's expected utility $U$ is the total number of users who choose $a_i=1$, which can be expressed as
\begin{equation}\label{U_basic}
\begin{array}{ll}
    U & = \mathbb E\bigg[\sum_i a_i\bigg]  = \sum_i (\mu + (1-\mu)\beta)p_{i\mathbf{P}_i},   
\end{array}
\end{equation}
if $\beta\le\beta'$.
Otherwise, the sender fails to persuade the users, and therefore $U=0$.

Before introducing social-network aspects of our model, let us consider a scenario with only one user and one platform, ${\mathbf{A}=\mathbf{P}_i=\mathbf{P}_S}$, and see the implications of the setting so far.
If the sender chooses $\beta\le\beta'$, their utility will be
$U = (\mu + (1-\mu)\beta)p_{i\mathbf{A}}.$
Since the utility function increases monotonically for $\beta\in[0,\beta']$, 
the sender's optimal strategy is $\beta^*=\beta'$.
They prefer a deceitful strategy as long as the user trusts them.

Suppose the platform sets a regulation $\rho_\mathbf{A}$ that prohibits the sender from using a $\beta$ value greater than $\rho_\mathbf{A}$.
This, for example, could be done by setting a content moderation guideline to de-platform a news source that reports misinformation often.
In the absence of an alternative platform, the sender must comply with the regulation and restrict their deceitfulness to the range $\beta\in[0,\rho_\mathbf{A}]$ instead of $\beta\in[0,1]$, because a de-platformed sender's utility would be zero.
A relatively strict regulation $\rho_\mathbf{A}<\beta'$ works effectively in reducing the sender's optimal strategy $\beta^*$ from $\beta'$ to $\rho_\mathbf{A}$, thus improving the quality of news information on the platform.
Conversely, a loose regulation with $\rho_\mathbf{A}\ge\beta'$ is meaningless, as the sender's optimal strategy remains at $\beta^*=\beta'$.

\subsection{Social Interaction and Platform Adoption}

Users and the sender are connected via an undirected graph, where edges represent friend relationships. Given choices $\mathbf{P}_i$'s and $\mathbf{P}_S$, this graph induces two subgraphs containing actors that are on the same platform. We assume that the signal originating from the sender travels along the edges of the subgraph corresponding to the sender's platfom choice and that user $i$ receives that message with probability $p_{i\mathbf{P}_i}=p^k$ if and only if $\mathbf{P}_i=\mathbf{P}_S$, and $i$ is $k$ edges away from $S$ in that subgraph. 
We call parameter $p$ information diffusiveness.
Works such as \cite{ediger2010massive} have shown that news dissemination follows tree-like broadcast patterns.
To account for this, we assume that the receiving probability depends on the shortest path length in the subgraph.


Note that the probability $p_{i\mathbf{P}_i}$ depends on other users' choices of platforms, and a user in platform $\mathbf{P}_i\neq \mathbf{P}_S$ does not receive the signal.
See Figure \ref{fig:illustration} for an illustrated example.

\begin{figure}
\begin{center}
\includegraphics[width=8.8cm]{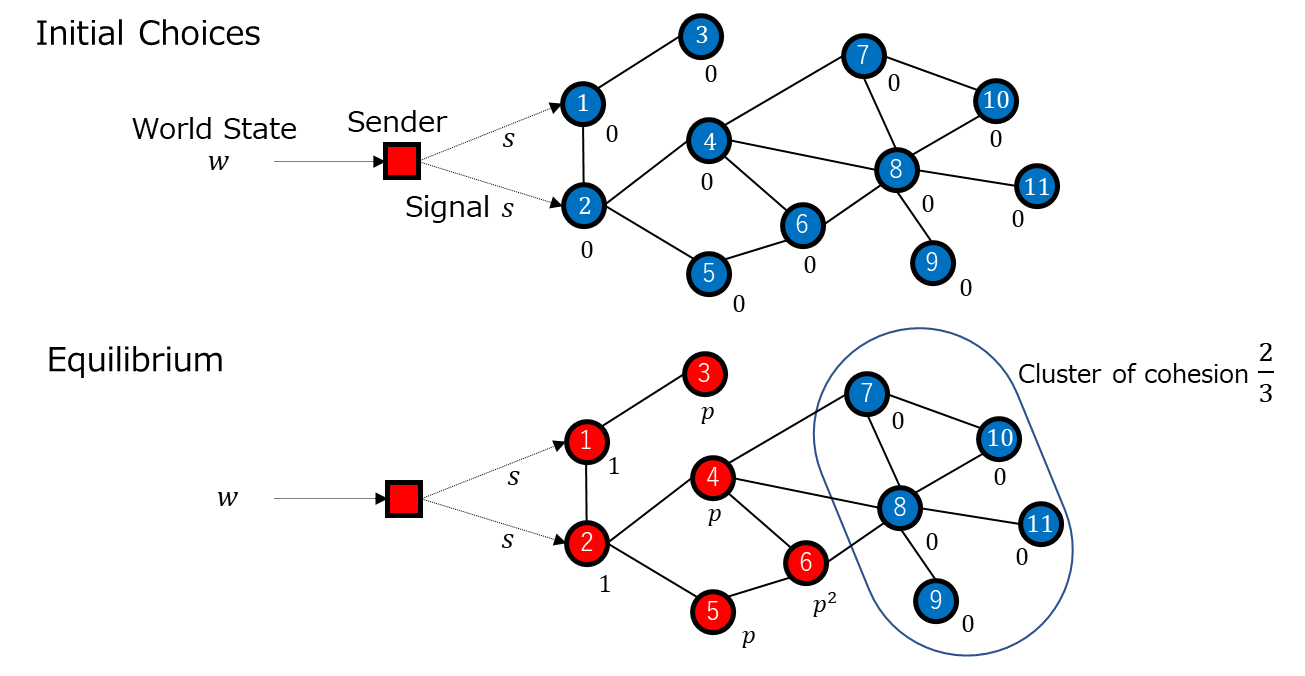}
\vspace{-.4cm}
\caption{Illustration of our model. Blue (red) nodes are in platform A (B). Values below nodes present $p_{i\mathbf{P}_i}$.
Note that user 6 is assumed to receive the source's signal with probability $p^2$ (because it is two edges away from it) even though there are two paths connecting these vertices. See text for details.
} 
\label{fig:illustration}
\end{center}
\end{figure}

The number of user $i$'s friends in platform $\mathbf{P}$ (i.e., the size of their neighborhood in the corresponding subgraph) is denoted by $N_{i\mathbf{P}}$. The payoff derived from social interactions with these friends is 
$\Phi_{i\mathbf{P}_i} = N_{i\mathbf{P}_i}b_{\mathbf{P}_i}$, where $b_\mathbf{P}$ is a parameter, controlled by $\mathbf{P}$, which determines the reward obtained for each social interaction in platform $\mathbf{P}$. The value of $b_\mathbf{P}$ should be thought of as increasing with every service or feature (e.g., Story/Fleet feature, birthday notification...) introduced by the platform to enhance a user's experience interacting with peers. 
We call $b_\mathbf{P}$ the quality of social interaction.
In real life, some friends may be more important than others, but the uniform $b_\mathbf{P}$ for all neighbors is a typical assumption in the literature on networked agents’ adoption of a new option with positive externality (e.g., \cite{Morris2000Contagion}).
All in all, user $i$'s utility in platform $\mathbf{P}$ is given by
\begin{equation}
\label{uti}
    V_{i\mathbf{P}} 
    = \Phi_{i\mathbf{P}} + \Psi_{i\mathbf{P}}.
\end{equation}
Given the sender's decision on $\mathbf{P}_S$ and $\beta\in[0,\beta']$, user $i$ should choose  $\mathbf{P}_i$ such that $V_{i\mathbf{P}_i}\ge V_{i\mathbf{P}}$ for all $\mathbf{P}\in\{\mathbf{A},\mathbf{B}\}$.
We say the users' choices of platforms are in equilibrium if this condition is satisfied for all users.
A trivial equilibrium is $ \mathbf{P}_i=\mathbf{P}_S$ for all $i$, but there can be other equilibria.

Since we are interested in situations where a pre-existing, originally ``dominant'', platform ($\mathbf{A}$) imposes regulations in the presence of a competing one ($\mathbf{B}$), it makes sense to specifically look for equilibria that occur as the limit of an ``adoption process'' carried out by the users. 
The process is described as follows.
First, the sender chooses $\mathbf{P}_S$.
Then, from the initial state with all users in $\mathbf{A}$, users repeatedly update their choices. 
In each iteration, users choose $\mathbf{P}_i=\operatorname{argmax}_\mathbf{P}{V_{i\mathbf{P}}}$ simultaneously, based on other users' previous choices.
Iterations continue until an equilibrium is reached, if any.
We will discuss the convergence of this process next, noting for now that it is essentially a fictitious play-iteration (see, e.g., \cite{Morris2000Contagion}) for the non-cooperative game played by the users with utilities $V_{i\mathbf{P}}$, once the sender’s choice is enacted.

\begin{prop}\label{converge_acyclic}
If 1) the social network is acyclic and platform $\mathbf{A}$ provides higher social interaction quality than platform B, i.e., $b_\mathbf{A} > b_\mathbf{B}$,
or 2) the social network is a finite graph,
then the adoption process converges to an equilibrium in a finite number of iterations. 
\end{prop} 

This is the sole equilibrium we will consider going forward, when saying that a property holds ``at (the) equilibrium.''

\begin{proof}
We will prove the result for Condition 1) only, as the proof for 2) is similar.
If $\mathbf{P}_S=\mathbf{A}$, the initial state is the equilibrium.
So we only need to consider $\mathbf{P}_S=\mathbf{B}.$

As shown in the Appendix, users can switch from platform $\mathbf{A}$ to $\mathbf{B}$ but not from $\mathbf{B}$ to $\mathbf{A}$.
It is also shown that for an acyclic graph, user $i$ that switches to platform $\mathbf{B}$ at iteration $t$ is $t$ edges away from the sender and has $p_{i\mathbf{B}}^{(t)}=p^t$, where $x^{(t)}$ denotes the value of $x$ at iteration $t$. (This implies that such user $i$'s neighbor who is $t-1$ edges away from the sender has to switch to $\mathbf{B}$ at iteration $t-1$. It  also implies that if a user $k$ edges away does not switch to $\mathbf{B}$ at iteration $k$, then the user and its downstream users do not switch to $\mathbf{B}$ forever.)

Suppose user $i$ switches to $\mathbf{B}$ at an iteration of sufficiently large $t$.
Since $p_{i\mathbf{A}}^{(t)}=0, p_{i\mathbf{B}}^{(t)}=p^t$, we have $\Psi_{i\mathbf{B}}^{(t)} < \Psi_{i\mathbf{A}}^{(t)} + \epsilon$ for small $\epsilon>0$.
Since user $i$'s neighbors farther from the sender are still in platform $\mathbf{A}$ and the neighbor closer to the sender is in $\mathbf{B}$,
$N_{i\mathbf{B}}^{(t)}=1$.
Since $b_\mathbf{B}<b_\mathbf{A}$, we have $\Phi_{i\mathbf{B}}^{(t)}<\Phi_{i\mathbf{A}}^{(t)}$
if $N_{i\mathbf{A}}^{(t)}\ge 1$.
Therefore, $V_{i\mathbf{B}}^{(t)}<V_{i\mathbf{A}}^{(t)}$
if $N_{i\mathbf{A}}^{(t)}\ge 1$.
Since user $i$ adopts $\mathbf{B}$ at iteration $t$ 
(i.e., $V_{i\mathbf{B}}^{(t)}\ge V_{i\mathbf{A}}^{(t)}$),
this means $N_{i\mathbf{A}}^{(t)}=0$.
Therefore, user $i$ is a leaf node.
Hence, no user switches to $\mathbf{B}$ at iteration $\tau \ge t+1$.

\end{proof}

\section{Strictest Effective Regulation \label{Sec3}}

We now focus on the scenario where platforms have the ability to regulate the level of deceitfulness of the sender.
Specifically, we assume that platform $\mathbf{A}$ enforces a restriction $\rho_\mathbf{A} \in [0,1]$ on the sender's choice of $\beta$, limiting the range to $[0,\rho_\mathbf{A}]$. In contrast, $\mathbf{B}$ has no such restriction, i.e., $\rho_\mathbf{B}=1$, and the sender can choose any $\beta\in[0,1]$ if they choose  $\mathbf{B}$. In practice, such regulation can be implemented if the source has been present on the platform long enough to enable it to fact-check its messages over a period of time, thus obtaining an empirical estimate of $\beta$.
\begin{definition}
\label{def_eff}
Platform $\mathbf{A}$’s regulation $\rho_\mathbf{A}$ is said to be \textit{effective} if it decreases the sender’s deceitfulness in platform $\mathbf{A}$ at equilibrium.
The strictest effective regulation $\rho_{SE}$ is the minimum of effective regulation $\rho_\mathbf{A}$.
\end{definition}

In other words, Definition \ref{def_eff} states that $\rho_\mathbf{A}$ is effective if in the presence of the constraint $\beta \in [0,\rho_\mathbf{A}]$, (1) the sender stays in $\mathbf{A}$ and (2) the chosen value $\beta^*$ is less than if $\rho_\mathbf{A}$ were equal to 1 (i.e., $\beta^* < \beta'$).
Our primary interest lies in identifying the strictest effective regulation $\rho_{SE}$. 
With this regulation, all users remain in platform $\mathbf{A}$.

Assuming that either condition 1) or 2) in Proposition \ref{converge_acyclic} holds, we let $U_\mathbf{P}$ be the sender's expected utility in platform $\mathbf{P}$.
When users take the optimal strategies, the sender considers $U_\mathbf{P}$ as a function of $\beta$.
We denote $U^*_\mathbf{P}$ as the maximum value of $U_\mathbf{P}$ for $\beta\in [0,\rho_\mathbf{P}]$.

If the sender chooses platform $\mathbf{A}$, all users should choose $\mathbf{A}$ at the equilibrium.
If $\sum_i p_{i\mathbf{A}}<\infty$, according to (\ref{U_basic}), $U_\mathbf{A}$ increases linearly with $\beta\in[0, \min\{\rho_\mathbf{A}, \beta'\}]$.
Thus, the sender's optimal level of deceitfulness is $\beta^*=\min\{\rho_\mathbf{A}, \beta'\}$.
On the other hand, if $\sum_i p_{i\mathbf{A}}=\infty$, then $U^*_\mathbf{A}=\infty$ for any $\beta\in[0, \min\{\rho_\mathbf{A}, \beta'\}]$.

This observation allows us to characterize an effective regulation.
Firstly, if $\rho_\mathbf{A}\ge\beta'$, then the regulation has no impact on the sender's behavior as the sender would have chosen a lower level of deceitfulness.
In other words, the regulation is too lenient to be effective.
Secondly, if regulation $\rho_\mathbf{A}<\beta'$ satisfies $U_\mathbf{A}(\rho_\mathbf{A})\ge U^*_\mathbf{B}$, then the sender should remain in platform $\mathbf{A}$ since it would not gain more utility in platform $\mathbf{B}$.
Note that when $\rho_\mathbf{A}$ takes the lowest value that satisfies the inequality, we have
\begin{equation}
    (\mu + (1-\mu)\rho_\mathbf{A})\sum_i p_{i\mathbf{A}} = U^*_\mathbf{B},
\end{equation}
where $p_{i\mathbf{A}}$ corresponds to $\mathbf{P}_S=\mathbf{A}$.
Lastly, if regulation $\rho_\mathbf{A}<\beta'$ satisfies $U_\mathbf{A}(\rho_\mathbf{A}) < U^*_\mathbf{B}$, then the sender should switch to $\mathbf{B}$ for higher utility.
The regulation is too strict, and the sender would prefer to move to the alternative platform with some of its followers.

Rearranging these cases from the viewpoint of designing the strictest effective regulation $\rho_{SE}$ in platform $\mathbf{A}$, we have the following proposition.

\begin{prop}\label{strictest_regulation}
If $U_\mathbf{A}(\beta')\le U^*_\mathbf{B}$, effective regulation does not exist.
If $U_\mathbf{A}(0) > U^*_\mathbf{B}$, platform $\mathbf{A}$ can enforce any strict regulation effectively, i.e., $\rho_{SE}=0$. 
If $U_\mathbf{A}(0) \le U^*_\mathbf{B} < U_\mathbf{A}(\beta')$, regulation should be moderate and 
\begin{equation}\label{strictest_regulation_formula}
    \rho_{SE}=\frac{1}{1-\mu}\bigg(\frac{U^*_\mathbf{B}}{\sum_i p_{i\mathbf{A}}}-\mu\bigg),
\end{equation}
where $p_{i\mathbf{A}}$ is for $\mathbf{P}_S=\mathbf{A}$.
\end{prop}

\section{Discussion of model assumptions}

Having now presented the main features of our model, and before particularizing its high-level predictions to specific networks, we are in a good position to discuss its central assumptions and simplifications in more details.

\subsection{Multiple platforms}
The main insights of our model remain relevant in situations where more than two platforms are competing, with one of them being initially dominant. In fact, the two-platform model can be considered as the worst case for this dominant platform because the key factor of platform competition is the positive externality. If users can disperse to many alternative platforms, it is difficult for each one to become competitive against the dominant one. Therefore, the dominant platform might be able to enforce stricter regulation.

\subsection{Multihoming}
In its current form, our model assumes that users adopt one platform at a time. In reality, users do not necessarily have to leave one platform to join another and can, a priori, consume news and interact with peers on multiple sites contemporaneously. If such multihoming occurs, the sender's departure to another platform may not necessarily cause users to leave the dominant one, and the latter can in turn enforce \textit{any} regulation without any risk of losing its user base.

However, if user $i$ has a finite time and/or attention budget to allocate between the various sites they join, and if they do not explicitly value diversity of platforms, they will de facto spend all that budget on the platform bringing them highest utility, i.e., move to $\mathbf{P}$ such that $V_{i\mathbf{P}} \geq V_{i\mathbf{P}'} $ for all $\mathbf{P}'$, as we have assumed in Section \ref{Sec3}.

\subsection{Sender strategy and multiple senders}

Using a similar argument as the one we described for multihoming, one can contend that users do not have to get their news from a single source. In turn, if a specific source leaves the dominant platform, a user can just either continue relying on other sources or decide to adopt a new one. For this reasoning to hold, however, it is necessary that all news sources be seen as substitutable by the users.  Assuming otherwise (as we do) and, in particular, that the news consumption payoff enters a user's utility according to (\ref{uti}) thus implicitly posits that (a) they have already committed to a specific sender of interest, to the exclusion of all others and, (b) that this commitment is independent from the messages sent by the source. \\
These two points are consistent with the behavior of at least some categories of social media users,  whose main motivations for following an influencer are ``entertainment'' and ``out of habit'',  as uncovered by recent research in the field of social media studies \cite{Croes}. Conditions (a) and (b) are also likely to hold more generally if the news source has unique features that make it especially appealing or salient to users, even before receiving any message from it (e.g., if the source is a particularly visible public figure or media outlet, or a charismatic influencers).
Another way in which the users' ex-ante devotion to the source (encapsulated by (a) and (b)) manifests itself in our model is in the assumption that they know the parameter $\beta$. This could be the result of users having interacted with the source for a long enough time, and hence ``knowing what to expect'' from it, and/or $\beta$ being one of the salient features of the source that make it attractive to the users.  This is also consistent with the literature on Bayesian Persuasion, which has been used before to model interactions between platforms and news consumers \cite{Kamenica2009BayesianP,egorov2020persuasion,Hebert}. Note, however, that users do not necessarily heed the source's signal but, rather, use their knowledge of $\beta$ to interpret the signals they receive.

%




\section{Applications to Prototypical Networks}

In this section, we apply the main result of Section {sec3} to prototypical network structures. 
To maintain analytical tractability and be able to interpret results qualitatively, we focus on simple network structures with some relevance to actual social media platforms.
In particular, we study linear networks, star-chain networks, regular trees, and stochastic block models for the following reasons.

We analyze a linear network as it presents an intuitive analytical representation of our model.
In many social media settings, news dissemination takes place with a tree-like hub-spoke structure, often embedded within a larger social interaction network \cite{himelboim2017classifying, ediger2010massive}.
A star-chain network is built on this structure and possesses cohesive blocking clusters, a crucial element for platform adoption as defined below. 
\begin{definition}
    A cluster, or a set of users, has cohesion $x$ if each member has at least fraction $x$ of their neighbors within the set. 
    The cohesion of a cluster is the maximum of such $x$.
\end{definition}
Without news consumption payoffs, platform adoption process would be governed by cluster cohesion.
As shown in \cite{Morris2000Contagion}, a ``blocking cluster,'' that has the highest cohesion in a graph would determine whether the adoption to the whole population occurs or not.
A regular tree network exhibits, in addition to cohesive blocking clusters, the exponentially spreading communication feature, making it a good benchmark to gauge the relative importance of both characteristics (especially when compared to the star chain).

Large literature (e.g., \cite{mcpherson2001birds}, \cite{currarini2009economic}) has demonstrated that homophily is encoded in common social networks.
Stochastic block model is a random network suited to express such networks with community structures.
Analysis of this parameterized random network is useful because players may not know the exact structure of a large network, and even if they do, computing exact solutions is prohibitively expensive.
Instead, an approximate solution based on a few primary factors may be more practical. 
Coordination games on stochastic block model have been explored in prior works (e.g., \cite{jackson2017}, \cite{parise2019}).


\subsection{Linear Network}


An infinite linear network has users $i=0,1,2,\ldots$ with edges between $i$ and $i+1$. 
The news source can send a signal to user 0. 
If the sender and user $0,\ldots,k$ are in the same platform, user $k$ receives a signal with probability $p^k$.
We will use index $k$ instead of $i$ to imply the user is $k$ edges away from the sender, although $k$ coincides with $i$ in the linear network.

If the sender opts for platform $\mathbf{A}$, all users remain in $\mathbf{A}$ since $V_{k\mathbf{A}}>V_{k\mathbf{B}}$.
The sender's utility function is 
\vspace{-0.2cm}
\begin{equation}\label{UA}
    U_\mathbf{A}(\beta) 
    = \sum_{k=0}^{\infty}(\mu +(1-\mu)\beta)p^k
    = \frac{\mu +(1-\mu)\beta}{1-p}.
\end{equation}

If $\mathbf{P}_S=\mathbf{B}$, some users may switch platforms.
Let user $K$ be the last user who switches to $\mathbf{B}$.
We have $V_{k\mathbf{A}}\le V_{k\mathbf{B}}$ for $k\le K$, and $V_{k\mathbf{A}}> V_{k\mathbf{B}}$ for $k> K$, where
\begin{equation}\label{VkAVkB}
\begin{array}{ll}
    &V_{k\mathbf{A}}=b_\mathbf{A} +(1-\mu)c, \\
    &V_{k\mathbf{B}}= b_\mathbf{B} + \mu p^k(1-c) + (1-\mu)(1-\beta p^{k})c \\
\end{array}
\end{equation}
for $k=K, K+1$.
The sender thus maximizes
\begin{equation}
\begin{array}{ll}
    U_\mathbf{B}(\beta) 
    &= \sum_{k=0}^{K(\beta)}(\mu +(1-\mu)\beta)p^k 
\end{array}
\end{equation}
with $K$ satisfying the condition
\begin{equation}\label{condK}
    V_{K+1,B} < b_\mathbf{A} +(1-\mu)c \le V_{K,B}.
\end{equation}
It is worth noting that if $b_\mathbf{A}\le b_\mathbf{B}$, then $K=\infty$.

Condition (\ref{condK}) can be rewritten equivalently in two ways,
\vspace{-0.2cm}
\begin{equation} \label{cond_f}
    F(K+1) < \beta \leq F(K),
\end{equation}
\begin{equation} \label{cond_g}
    G(\beta)-1 < K \leq G(\beta),
\end{equation}
\vspace{-0.2cm}
where
\vspace{-0.2cm}
\begin{equation}\label{FG}
\begin{array}{ll}
    F(K) 
    &= \frac{1}{(1-\mu)c} \bigg(\mu(1-c)-\frac{b_\mathbf{A}-b_\mathbf{B}}{p^K} \bigg), \\
    G(\beta) 
    &= \frac{1}{\log p} 
    \log \bigg(\frac{b_\mathbf{A}-b_\mathbf{B}}{\mu(1-c)-(1-\mu)\beta c} \bigg).
\end{array}
\end{equation}
Note that $F(K), G(\beta)$ are decreasing functions, and the equalities in the conditions hold at the equilibrium.
We thus have 
\begin{equation}\label{Ustar_B}
\begin{array}{ll}
    U_\mathbf{B}^* 
    &= \max_{K\in \mathbb N} 
          (\mu+(1-\mu)F(K))\frac{1-p^{K+1}}{1-p} \\
    &=\max_{\beta: G(\beta)\in \mathbb N} 
          (\mu+(1-\mu)\beta)\frac{1-p^{G(\beta)+1}}{1-p}.
\end{array}
\end{equation}

Using these $U_\mathbf{A}(\beta)$ and $U^*_\mathbf{B}$, the strictest effective regulation can be designed.
In particular, \eqref{strictest_regulation_formula} becomes
\begin{equation}
    \rho_{SE}=\frac{(1-p)U^*_\mathbf{B}-\mu}{1-\mu}.
\end{equation}

The application of Proposition \ref{strictest_regulation} to this situation is illustrated in Figure \ref{fig:utility_linear}.
For all subsequent numerical experiments, we use the following parameter values unless otherwise specified: $\mu=0.2, c=0.3, p=0.9, b_\mathbf{A}-b_\mathbf{B}=0.01$.
(The value of $\mu$ does not affect the claims of this article qualitatively, and we will study different values of the other parameters in the following sections.)
When the sender chooses platform $\mathbf{B}$, the utility attains maximum $U^*_\mathbf{B}$ at $\beta=\beta^*_\mathbf{B}$.
If the sender selects $\beta>\beta^*_\mathbf{B}$, the signal is biased towards the sender's preference, but few users move to  $\mathbf{B}$ due to the poor news quality.
If $\beta<\beta^*_\mathbf{B}$, many users switch to $\mathbf{B}$ to follow the sender, but the signal is not biased in favor of the sender.
On the other hand, when the sender chooses $\mathbf{A}$, all users remain in the platform.
The sender's utility, $U_\mathbf{A}$, increases for $\beta\le \min\{\rho_\mathbf{A}, \beta'\}$, and therefore achieves the maximum value at $\beta=\min\{\rho_\mathbf{A}, \beta'\}$.
Proposition \ref{strictest_regulation} indicates, with these parameters, the strictest effective regulation is $\rho_{SE}$ such that $U_\mathbf{A}(\rho_{SE})=U^*_\mathbf{B}$.
If platform $\mathbf{A}$ enforces a stricter regulation $\rho_\mathbf{A}<\rho_{SE}$, then the sender's utility becomes lower than in $\mathbf{B}$.
In other words, the sender would rather send more biased information to fewer people in the alternative platform than comply with the regulation to gain more potential receivers, and thus platform $\mathbf{A}$ would lose some users.
On the other hand, regulation $\rho_\mathbf{A}\in[\rho_{SE}, \beta']$ works effectively, making the sender less deceptive, and the platform does not lose any users.

\begin{figure}
\begin{center}
\includegraphics[width=7.0cm]{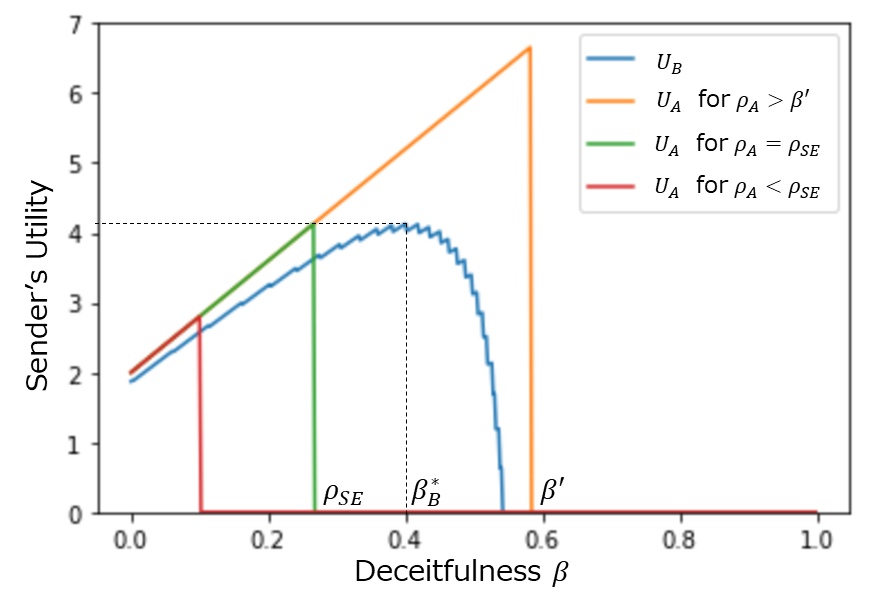}    
\vspace{-0.2cm}
\caption{The sender's expected utility in a linear network} 
\label{fig:utility_linear}
\end{center}
\end{figure}

We will now inspect for what $b_\mathbf{A}, b_\mathbf{B}, p$, platform $\mathbf{A}$ can enforce the strict regulation $\rho_\mathbf{A}=0$.
Our previous work \cite{sasaki2022necsys} proved a sufficient condition for $\rho_{SE}=0$:

\begin{prop}\label{prop_strict_regulation_linear}
    \textcolor{black}{\cite{sasaki2022necsys}} \;
    For an infinite linear network, 
    if $b_\mathbf{A}-b_\mathbf{B}>\mu(1-c)$ or $b_\mathbf{A}-b_\mathbf{B}>\mu(1-c)^2/p$, then the strictest effective regulation is $\rho_{SE}=0$. 
\end{prop}

We will improve this proposition with a tighter bound.

Let $b_\mathbf{A}\ge b_\mathbf{B}\ge 0$.
Suppose that the sender chooses $\mathbf{B}$. 
At the equilibrium, the sender's best strategy satisfies $\beta^*_\mathbf{B} = F(K)$.
The strictest effective regulation $\rho_{SE}$ takes the value of zero if and only if $U_\mathbf{A}(\beta=0)\ge U_\mathbf{B}(\beta^*_\mathbf{B})$, i.e.,
\vspace{-0.2cm}
\begin{equation}
    \sum_{k=0}^{\infty}\mu p^k \ge \sum_{k=0}^{K} (\mu +(1-\mu)\beta^*_\mathbf{B})p^k,
\end{equation}
or 
\vspace{-0.2cm}
\begin{equation}\label{linear_thresh_rho0}
    b_\mathbf{A} - b_\mathbf{B} \ge \mu p^{K} - \frac{\mu c p^{K}}{1-p^{K+1}} =: f_{K}(p).
\end{equation}
Therefore, using $f_{K}(p)$ defined here, we have the following.

\begin{prop}\label{bA_threshold_for_strictest_restriction}
For an infinite linear network, the strictest effective regulation is $\rho_{SE}=0$ if $b_\mathbf{A}-b_\mathbf{B} \ge \max_K f_{K}(p)$.
\end{prop}

In Figure \ref{fig:regulation}(a), the strictest effective regulation is calculated for different $p$ and $b_\mathbf{A}$, fixing $b_\mathbf{B}=0$.
The white curve presents $b_\mathbf{A} = \max_K f_{K}(p) + b_\mathbf{B}$ as in Proposition \ref{bA_threshold_for_strictest_restriction}. 
Above this curve, platformer $\mathbf{A}$ can enforce any regulation without losing users.

\begin{figure*}
\begin{center}
\includegraphics[width=0.86\textwidth]{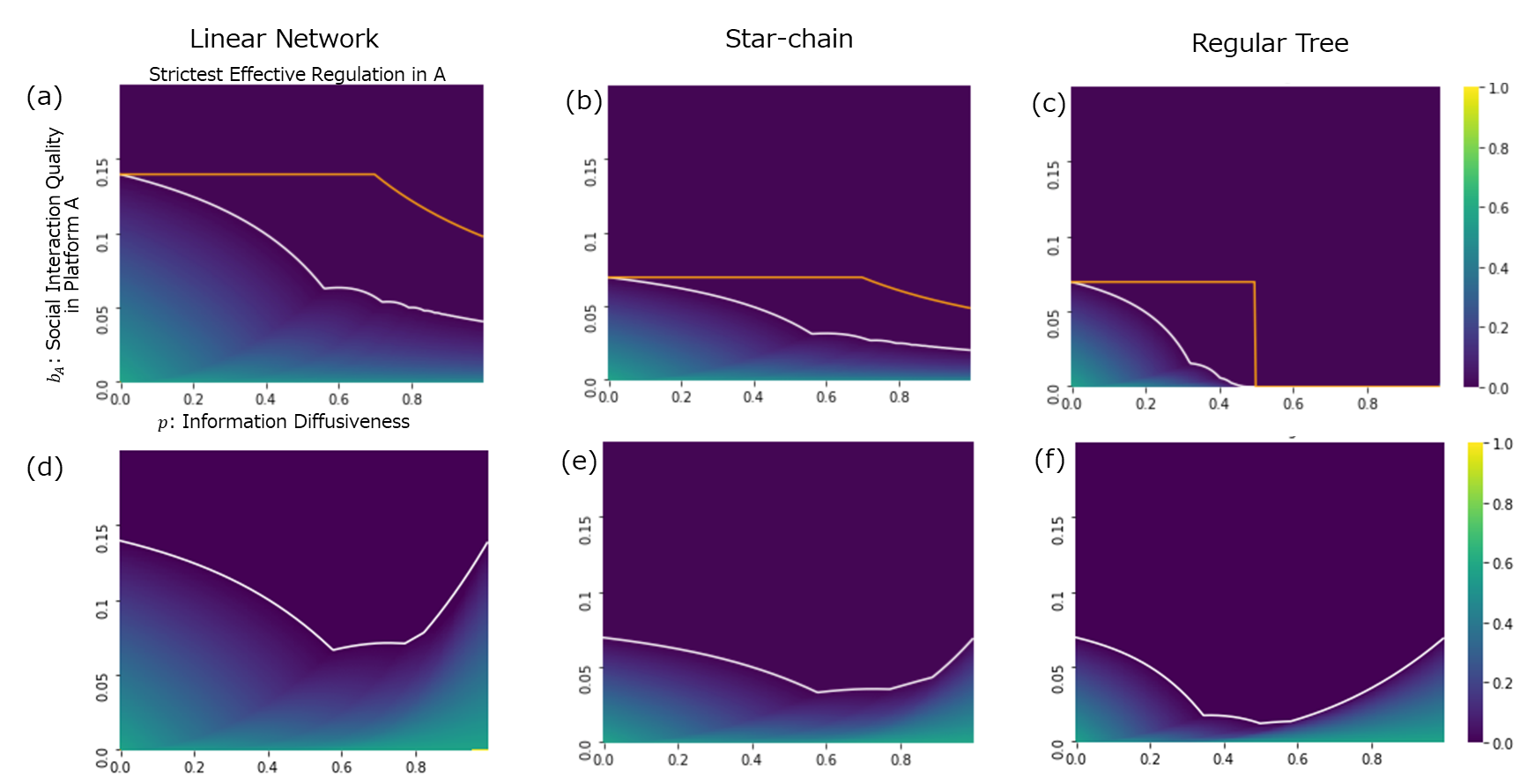}
\vspace{-.4cm}
\caption{
The strictest effective regulation $\rho_{SE}$ in (a) linear network, (b) star-chain network, and (c) regular-tree network. Respectively, (d), (e), (f) are the finite cases.
In panel (a), the white curve is computed using the formula of Proposition \ref{bA_threshold_for_strictest_restriction}.
From this proposition, the strictest effective regulation is guaranteed to be $\rho_{SE}=0$ above this curve in the $(p,b_\mathbf{A})$-plane. 
The orange curve plays the same role, but based on the results of Proposition \ref{prop_strict_regulation_linear}.
The color map is computed by \eqref{Ustar_B} and Proposition \ref{strictest_regulation}.
The white/orange curves and color maps in other panels are computed similarly.
}
\label{fig:regulation}
\end{center}
\end{figure*}


We now consider a finite linear network with users $0,1,\ldots,n-1$.
Note that $K$ does not take the value of $n-2$ since user $n-1$ always chooses the same platform as user $n-2$.
Following the argument in Appendix, we let
\vspace{-0.1cm}
\begin{equation*}
\begin{array}{ll}
     f_{n, K}(p):=\mu p^{K} - (1-p^n) \frac{\mu c p^{K}}{1-p^{K+1}}&  \text{for } K<n-2\\
     f_{n, n-2}(p):=\mu (1-c) p^{n-2}&  \text{for } K=n-1.
\end{array}
\end{equation*}
Then Proposition \ref{bA_threshold_for_strictest_restriction} becomes
\begin{prop}\label{bA_threshold_for_strictest_restriction_finite}
For a finite linear network, the strictest effective regulation is $\rho_{SE}=0$ if $b_\mathbf{A}-b_\mathbf{B} \ge \max_{K=0,1,\ldots,n-2} f_{n,K}(p)$.
\end{prop}

Panel (d) in Figure \ref{fig:regulation} shows the strictest effective regulation for a finite linear network.
As a rule, and as is evident from the figure, in the case of linear networks, platform $\mathbf{A}$ can impose any regulation on the sender as long as its quality of social interaction is sufficiently larger than that of platform $\mathbf{B}$.

\subsection{Star-chain Network}

In this section, we consider a star-chain network, 
where hub user $k$ has $r-1$ leaf nodes and two hub users $k-1, k+1$ as their neighbors. 
The news source sends the signal to hub user 0.
The cohesion of blocking clusters is $\frac{r}{r+1}$.
If hub users $0,1,\ldots,k$ are in the same platform as the sender, hub user $k$ receives the signal with probability $p^k$ (and their leaf nodes receive with probability $p^{k+1}$).

The news source's utility in platform $\mathbf{A}$ is 
\vspace{-0.2cm}
\begin{equation}
    U_\mathbf{A}(\beta) 
    = \sum_{k=0}^{\infty}(\mu +(1-\mu)\beta)(p^k + (r-1)p^{k+1}).
\end{equation}

Suppose the sender chooses platform $\mathbf{B}$.
Let $K$ be defined as an integer such that at the equilibrium, hub user $k\le K$ chooses platform $\mathbf{B}$ and hub user $k>K$ chooses $\mathbf{A}$.
Note that a peripheral user must always follow their hub user's choice of platform.
The sender in $\mathbf{B}$ maximizes their expected utility
\begin{equation}
\begin{array}{ll}
    U_\mathbf{B}(\beta) 
    = \sum_{k=0}^{K(\beta)}(\mu +(1-\mu)\beta)(p^k + (r-1)p^{k+1})
\end{array}
\end{equation}
subject to \footnote{For hub user $k=K, K+1$, consider
\begin{equation*}
\begin{array}{ll}
    &V_{k\mathbf{A}}= rb_\mathbf{A} +(1-\mu)c, \\
    &V_{k\mathbf{B}}= b_\mathbf{B} + \mu p^k(1-c) + (1-\mu)(1-\beta p^{k})c.
\end{array}
\end{equation*}}
\vspace{-0.2cm}
\begin{equation}\label{condK_starchain}
    V_{K+1,B} < rb_\mathbf{A} +(1-\mu)c \le V_{K,B}.
\end{equation}

Unlike linear network cases, even if the social interaction quality in platform $\mathbf{A}$ is lower than $\mathbf{B}$, users may prefer the former as it enables them to communicate with many friends.
Indeed, Proposition \ref{bA_threshold_for_strictest_restriction} becomes
\begin{prop}\label{bA_threshold_for_strictest_restriction_starchain}
For an infinite star chain, the strictest effective regulation is $\rho_{SE}=0$ if $rb_\mathbf{A}-b_\mathbf{B} \ge \max_K f_K(p)$.
\end{prop}

The proof follows the same argument as the linear network.
The proposition indicates that platform $\mathbf{A}$ needs relatively low quality of social interactions to enforce strict regulations when hub users have many peripheral users.
This is because users with many friends have high inertia to switch platforms.
Since users have more friends in platform $\mathbf{A}$ and higher incentive to stay there, the dominant platform has more power to affect players' behaviors.
The higher the cohesion of blocking clusters is, the easier it becomes to enforce strict regulations.

Panel (b) in Figure \ref{fig:regulation} presents numerical results for $r=2, b_\mathbf{B}=0$.
Compared to panel (a), the white curve is lower in the $(p,b_\mathbf{A})$-plane, and the dark area is larger in this panel.
This means that platform $\mathbf{A}$ is more capable of enforcing the strict regulation $\rho_\mathbf{A}=0$ than in the linear network case.

For a finite star-chain with $n$ hub users, the counterpart of Proposition \ref{bA_threshold_for_strictest_restriction_finite} also holds with some modification.\footnote{The farthest hub user can stay at platform $\mathbf{A}$ even when the other hub users are in $\mathbf{B}$, while in a finite linear network, the farthest user always chooses the same platform as the second farthest user.}
Letting 
\vspace{-0.2cm}
\begin{equation}
    g_{n, K}(p):= \mu p^{K} - (1-p^n) \frac{\mu c p^{K}}{1-p^{K+1}}
\end{equation}
for $K\le n-1$, Proposition \ref{bA_threshold_for_strictest_restriction_finite} becomes
\begin{prop}
For a finite star chain, the strictest effective regulation is $\rho_{SE}=0$ if $r b_\mathbf{A}-b_\mathbf{B} \ge \max_{K=0,1,\ldots,n-1} g_{n,K}(p)$.
\end{prop}
Panel (e) in the figure presents the result for $n=5$.

\subsection{Regular Tree Network}

In an $r$-regular tree network, each user is connected to $r$ child nodes. 
The news source sends the signal to the root node (user 0).
The blocking cluster cohesion is $\frac{r}{r+1}$ like for the star-chain network, but as the distance from the sender increases, the number of users grows exponentially. 
If users $0,1,\ldots,k$ along a path are in the same platform as the sender, then the probability that user $k$ receives the signal is $p_{k\mathbf{P}_k}=p^k$.

In platform $\mathbf{A}$, the sender's expected utility is 
\vspace{-0.2cm}
\begin{equation}
\begin{array}{ll}
    U_\mathbf{A}(\beta) 
    &= \sum_{k=0}^{\infty}(\mu +(1-\mu)\beta)p^k r^k.
\end{array}
\end{equation}

The sender in platform $\mathbf{B}$ maximizes the utility
\vspace{-0.2cm}
\begin{equation}
\begin{array}{ll}
    U_\mathbf{B}(\beta) 
    &= \sum_{k=0}^{K(\beta)}(\mu +(1-\mu)\beta)p^k r^k
\end{array}
\end{equation}
subject to condition (\ref{condK_starchain}). 
Proposition \ref{bA_threshold_for_strictest_restriction_starchain} holds when replacing function $f_K$ by $h_K$,
\vspace{-0.2cm}
\begin{equation}
    h_K(p) := \mu p^{K} - \frac{\mu c p^{K}}{1-(pr)^{K+1}}.
\end{equation}
In addition, we can show the following.
\begin{prop}\label{tree_infinite}
For an infinite regular tree, if
$rb_\mathbf{A}>b_\mathbf{B}$ and $pr>1$, then $\rho_{SE}=0$.
\end{prop}

\begin{proof}
    When $pr>1$, the expected number of users who receive the sender's signal is infinite.
    On the other hand, with $rb_\mathbf{A}>b_\mathbf{B}$, at most finite users move to platform $\mathbf{B}$.
    Therefore, the sender's utility is infinity in platform $\mathbf{A}$ and finite in $\mathbf{B}$. 
    By Proposition \ref{strictest_regulation}, since $U_\mathbf{A}(0)>U^*_\mathbf{B}$, 
    the sender should always choose platform $\mathbf{A}$.
    Any strict regulation thus works effectively, i.e., $\rho_{SE}=0$. 
\end{proof}

When a signal is diffusive, even lower quality of social interaction is sufficient for strict regulation because of exponentially many potential receivers.
Distant users are relatively important to the sender, and the dominant platform becomes powerful.
As long as the social interaction quality satisfies $rb_\mathbf{A}>b_\mathbf{B}$, platform $\mathbf{A}$ can enforce any regulation if the signal's reproduction rate $rp$ is greater than one.
Lower degree trees require higher $p$ for strict regulation.

Panel (c) in Figure \ref{fig:regulation} shows the strictest effective regulation for $r=2, b_\mathbf{B}=0$.
In the regular tree (c) and the star chain (b), platform $\mathbf{A}$ is more capable of enforcing strict regulation than in the linear network (a), due to cohesive blocking clusters.
In the regular tree (c), platform $\mathbf{A}$ is even more capable than in the star chain (b) due to the number of the potential receivers.
If $p>1/r$, a regular tree has $\rho_{SE}=0$ for any $b_\mathbf{A}\ge 0$.

The finite case follows the same argument as Proposition \ref{bA_threshold_for_strictest_restriction_finite}.
Let
\vspace{-0.2cm}
\begin{equation*}
\begin{array}{ll}
    h_{n, K}(p) := \mu p^{K} - (1-(pr)^n) \frac{\mu c p^{K}}{1-(pr)^{K+1}}& \text{for }K<n-2, \\
    h_{n, n-2}(p) := \mu (1-c) p^{n-2}& \text{for }K=n-1.
\end{array}    
\end{equation*}
\begin{prop}\label{bA_threshold_for_strictest_restriction_tree}
For a finite regular tree, the strictest effective regulation is $\rho_{SE}=0$ if $rb_\mathbf{A}-b_\mathbf{B} \ge \max_{K=0,1,\ldots,n-2} h_{n,K}(p)$.
\end{prop}
Note that Proposition \ref{tree_infinite} does not hold since the signal does not spread to infinitely many users.
Panel (f) shows the result for a finite tree truncated at the fifth generation.

\subsection{Stochastic Block Model}

\begin{figure*}
\centering
\includegraphics[width=1\textwidth]{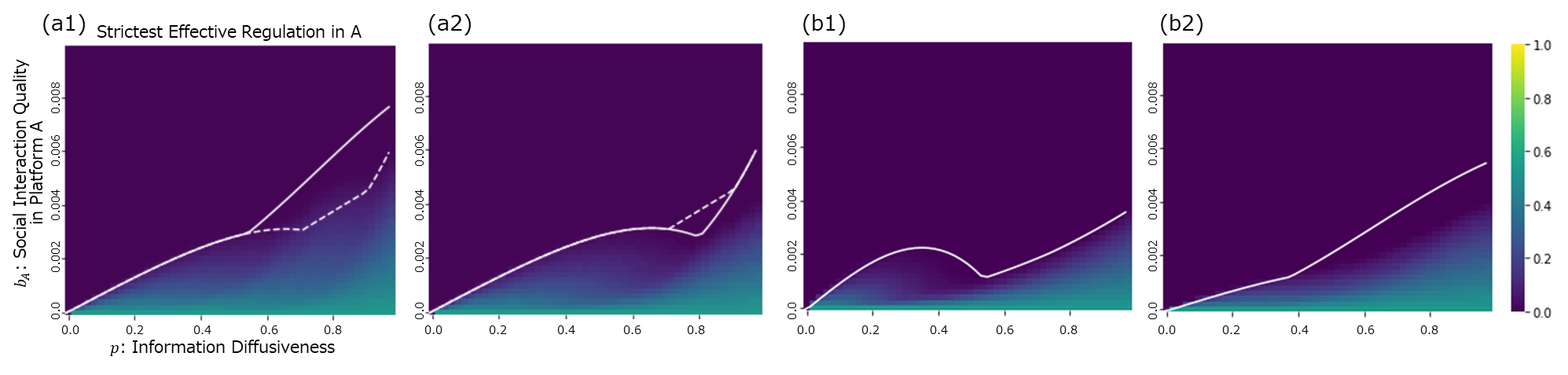}
\vspace{-.8cm}
\caption{
The strictest effective regulation $\rho_{SE}$ for (a1, a2) community chains and (b1, b2) complete graphs of communities.
The community chain has a middle community of (a1) low or (a2) high intra-community tightness.
Parameter values are as described in the main text. 
The color map is the simulation results of 50 random graphs.
The white curves calculated by Proposition \ref{thm_sbm_thresh} fit the results of the simulation, which does not use Assumptions \ref{assumption_community_migration}--\ref{assumption_external_friend} for the proposition.
The dashed curves represent the base case with medium intra-community tightness (see Figure \ref{fig:regulation_differentc_communitychain} panel (b)).
With high $p$, platform $\mathbf{A}$ can enforce strict regulation more easily (i.e., the white curve is lower) in (a2) than (a1).
For the complete graphs of communities, the sender is in a (b1) small or (b2) large community.
The white curves calculated by Proposition \ref{thm_sbm_thresh} again fit the color maps (the simulation results for 50 random graphs), which are computed without Assumptions \ref{assumption_community_migration}--\ref{assumption_external_friend} for the proposition.
With low $p$, platform $\mathbf{A}$ can enforce strict regulation more easily in (b2) than (b1).
}
\label{fig:sbm}

\end{figure*}

A stochastic block model is a random graph with $M$ communities.
Community $J\in \{1,\ldots, M\}$ has $n_J$ users, and a pair of users in communities $J, J'$ are friends with probability $\theta_{JJ'}\in [0,1]$.

To provide analytical expressions, we restrict our study to stochastic block models on which the platform reallocation process satisfies the few assumptions mentioned below. 

\begin{assumption}\label{assumption_community_migration}
    If a user moves to platform $\mathbf{B}$, then the other users in the same community do so too.
\end{assumption}
While this assumption may appear restrictive, we show through numerical simulations in Appendix that it is satisfied for some stochastic block models when $\theta_{JJ}$ is close to one.

\begin{assumption}\label{assumption_migration_order}
    Communities move to platform $\mathbf{B}$ in the order of community $1,2,\ldots,M$ (if they relocate).
\end{assumption}
We restrict our attention to network structures such that we can estimate the order of migration.
Such network class includes various community structures, for instance a chain, a star graph, and a complete graph of communities as we will see in the following subsections.

\begin{assumption}\label{assumption_external_friend}
    Friendship across communities is rare (i.e., $\theta_{JJ'}\approx 0, J\neq J'$), and the first-relocating user in a community has a single friend in the precedingly-relocated communities.
\end{assumption}
The argument can be naturally extended to more connected first-relocating users, provided that the number of such external friends, not necessarily one, is known to us.
Define user $\phi(J)$ as the first (potentially) relocating user in community $J$.

The following proposition is derived in the Appendix.
\begin{prop}\label{thm_sbm_thresh}
    Under Assumptions \ref{assumption_community_migration}--\ref{assumption_external_friend}, the strictest effective regulation is ${\rho_{SE}=0}$ if 
\vspace{-0.2cm}
\begin{equation*}
    n_J \theta_{JJ}b_\mathbf{A}-b_\mathbf{B} 
    \ge \mu (1-c)p_{\phi(J)\mathbf{B}} - \frac{\sum_{L=J+1}^M R_L}{\sum_{L=1}^J R_L} \mu c p_{\phi(J)\mathbf{B}}.
\end{equation*}
    for all $J=1,\ldots,M$, where $R_J$ denotes the expected number of signal receivers in community $J$ when the community is in the same platform as the sender.
\end{prop}
Proposition \ref{thm_sbm_thresh} emphasizes the role played by $p_{\phi(J)\mathbf{B}}$ and $R_J$ in enabling the platformer to apply regulation $\rho_{SE}=0$. These aggregate numbers can further be related to the basic structure and characteristics of the network if, e.g., they satisfy particular scaling laws. In that case, Proposition \ref{thm_sbm_thresh} can be used to gauge the effect of such parameters as well.

\subsubsection{Community-chain-type Scaling}

Our first example is $p_{\phi(J)\mathbf{B}}$ and $R_J$ decaying with $J$ according to
\vspace{-0.2cm}
\begin{equation}\label{communitychain_condition}
\begin{aligned}
    p_{\phi(J)\mathbf{B}} &= \begin{cases}
                  p  & J=1 \\
                  p^{2J-2} & 2\le J\le M,
              \end{cases} \\
    R_J &= n_J p^{2J-1}.
\end{aligned}
\end{equation}

As explained heuristically in the Appendix, such a scaling can be achieved in a so-called ``chain of communities'' network. 
Numerical simulations also presented in Appendix seem to indicate that Assumptions \ref{assumption_community_migration}--\ref{assumption_external_friend} hold for these networks when the intra-community tightness (i.e., probability $\theta_{JJ}$) is high, which makes them good representatives of the class of networks for which \eqref{communitychain_condition} holds and Proposition \ref{thm_sbm_thresh} is applicable.
Note that this is all that is needed of the network to perform our analysis. 


For example, consider a base model with $n_1=n_2=n_3=30$ and $\theta_{JJ'}$ given by
\begin{equation}\label{community_chain_theta}
    \theta = 
    \begin{pmatrix}
        \frac{3}{4}   & \frac{4}{900} & 0 \\
        \frac{4}{900} & \frac{3}{4}   & \frac{4}{900} \\
        0             & \frac{4}{900} & \frac{3}{4}   \\
    \end{pmatrix}.
\end{equation}

We utilize Proposition \ref{thm_sbm_thresh} to analyze how the regulation depends on the in-between community's ``tightness.''
For variants of the basic parameter sets, instead of $\theta_{22}=3/4$ for the second community, we consider $\theta_{22}=1/2$ (case (a1)) and $\theta_{22}=1$ (case (a2)).

Figure \ref{fig:sbm} shows the strictest regulation $\rho_{SE}$ for these networks.
The simulation does not use Assumptions \ref{assumption_community_migration}-\ref{assumption_external_friend} a priori, which Proposition \ref{thm_sbm_thresh} is based on.
The white curves indicate that our mathematical analysis of Proposition \ref{thm_sbm_thresh} fits the simulation results (the color map) of the corresponding physical model.
In particular, with high $p$, platform $\mathbf{A}$ can enforce strict regulation more easily for (a2), featuring a tightly connected intermediate community, than for (a1), which contains a loosely connected one.
This is in agreement with our qualitative understanding. 
When $p$ is high, the sender's influence on distant users becomes relatively important. 
Since a tightly connected intermediate community blocks their platform migration, it helps platform $\mathbf{A}$ hold more power to impose regulation.

\subsubsection{Community-star-type Scaling}

Another example is given by $p_{\phi(J)\mathbf{B}}$ and $R_J$ that, in contrast, do not decay with $J$:
\begin{equation}\label{communitystar}
\begin{aligned}
    p_{\phi(J)\mathbf{B}} &= \begin{cases}
                  p  & J=1 \\
                  p^2 & 2\le J\le M,
              \end{cases} \\
    R_J &= \begin{cases}
                  n_1 p  & J=1 \\
                  n_J p^3 & 2\le J\le M.
              \end{cases}
\end{aligned}
\end{equation}
As explained again in the Appendix, such scalings can be obtained for a star and a complete graph of communities (while satisfying Assumptions \ref{assumption_community_migration}--\ref{assumption_external_friend}), when the sender is in community $1$ and $\theta_{1J}>0$ for all $J=2,\ldots, M$.
We assume, without loss of generality, $n_2\theta_{22} \le n_3\theta_{33} \le \ldots \le n_M\theta_{MM}$.
This implies that the sender can attract the first users in community $2,\ldots,J-1$ more easily than community $J$ since they have less friends within community.
Assumption \ref{assumption_migration_order} is thus satisfied.

Simulations are conducted for complete graphs of communities with the sender at different communities.
The model has four communities.
In one simulation setting, the sender is in a small community, and $n_1=20, n_2=30, n_3=40, n_4=50$.
In another, the sender is in a large community, and $n_1=50, n_2=20, n_3=30, n_4=40$.
The other parameters are
\begin{equation}\label{completegraph_theta}
    \theta = 
    \begin{pmatrix}
        \frac{3}{4}      & \frac{4}{n_1n_2} & \frac{4}{n_1n_3} & \frac{4}{n_1n_4} \\
        \frac{4}{n_2n_1} & \frac{3}{4}      & \frac{4}{n_2n_3} & \frac{4}{n_2n_4} \\
        \frac{4}{n_3n_1} & \frac{4}{n_3n_2} & \frac{3}{4}      & \frac{4}{n_3n_4} \\
        \frac{4}{n_4n_1} & \frac{4}{n_4n_2} & \frac{4}{n_4n_3} & \frac{3}{4}      \\
    \end{pmatrix}
\end{equation}
so that the network structure itself remains identical in the two settings.

In Figure \ref{fig:sbm} (b1, b2), the white curves indicate that our mathematical analysis of Proposition \ref{thm_sbm_thresh} is again consistent with the simulation results, which do not presume the Assumptions \ref{assumption_community_migration}--\ref{assumption_external_friend}.
With low $p$, platform $\mathbf{A}$ can enforce strict regulation more easily if the sender is in a large community (panel (b2)) than if it is in a small one (panel (b1)).
This goes along with our qualitative intuition.
When $p$ is low, the sender's influence on nearby agents becomes relatively important. 
Since it is difficult for the sender to persuade a large community into platform $\mathbf{B}$, platform $\mathbf{A}$ can take advantage of this to impose regulation.

\section{Influencer and Neighboring Sympathizers}

The previous section demonstrated  how the strictest effective regulation is affected by network structure. In particular, we saw how the presence of cohesive clusters can preclude new platform adoption (and hence help the platformer impose stricter regulation), and how the location of pivotal users (who determine the power balance between the sender and the platformer) changes depending on information diffusiveness.

In this section, we consider yet another effect, namely heterogeneity in the users' preference for their received signal. More precisely, from now on, we assume users differ in how much they appreciate the news information received from the sender and in how biased they are toward unorthodox views.

This is captured by the payoff $c_i$ of user $i$ (replacing the common value $c$ used up to now).
A user with low $c_i$ gets high utility when they correctly estimate the surprising world state, i.e., $a_i=w=1$.
User $i$ with high $c_i$ gets high utility when they correctly estimate the unsurprising world state, i.e., $a_i=w=0$.
Therefore users with low $c_i$, having high $\beta'$, tend to trust a deceitful sender and bet on the unlikely world state.
Note that with the standing assumption of $0<\mu<c_i<1/2$, their default estimation is still $a_i=0$. 
We call users with low $c_i$ \textit{sympathizers}, and high $c_i$ \textit{non-sympathizers}.

\begin{figure*}
\begin{center}
\includegraphics[width=0.83\textwidth]{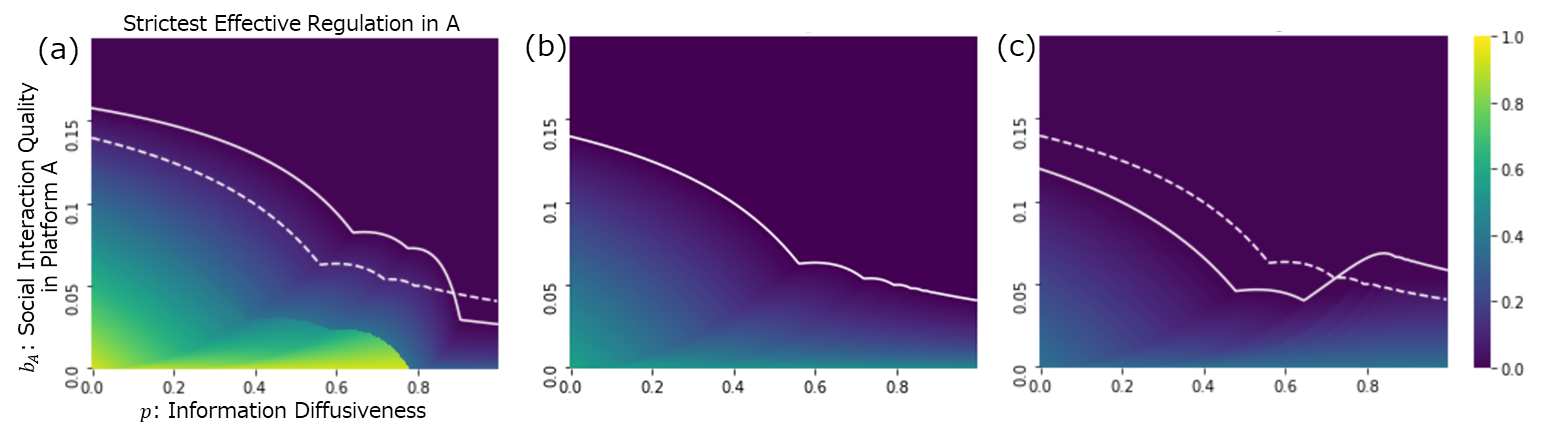}
\vspace{-.4cm}
\caption{The strictest effective regulation $\rho_{SE}$ for infinite linear networks. Panel (b) is the basic case with $c_i=0.3$ for all $i$.
For panel (a), users close to the sender are sympathizers ($c_i=0.21$ for $i=0,1,2$) and distant users are non-sympathizers ($c_i=0.4$ for $i\ge 3$).
Panel (c) is the opposite ($c_i=0.4$ for $i=0,1,2$ and $c_i=0.21$ for $i\ge 3$).
The white curves are the parameter threshold for $\rho_{SE}=0$ as in Proposition \ref{bA_threshold_for_strictest_restriction}, and the dashed curves represent that in the base case (b) for reference. 
In case (a), it is difficult for platform $\mathbf{A}$ to enforce strict regulation (i.e., the color map is bright and the white curve is high) when $p$ is low.
The color map is computed similarly to Figure \ref{fig:regulation}, not using any random procedure.} 
\label{fig:regulation_differentc_linear}



\includegraphics[width=0.83\textwidth]{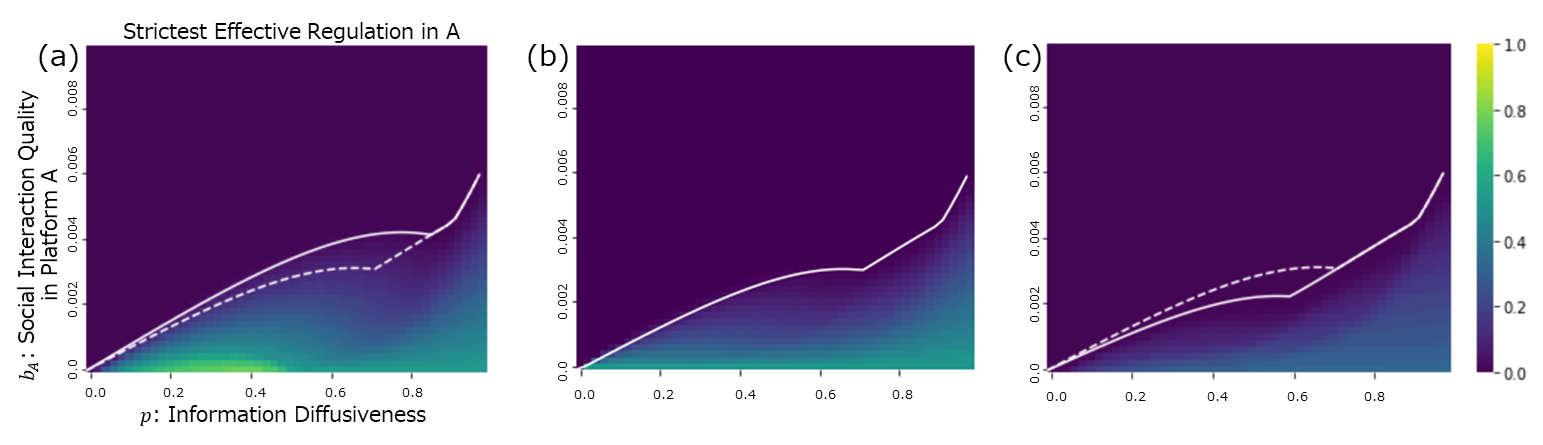}
\vspace{-.4cm}
\caption{The strictest effective regulation $\rho_{SE}$ for a community chain. Panel (b) is the basic case with $c_J=0.3$ for all $J=1,2,3$.
For panel (a), users in the sender's community are sympathizers ($c_1=0.21$) and distant users are the same as the basic case ($c_2=c_3=0.3$).
Panel (c) has non-sympathizers nearby ($c_1=0.4$ and $c_2=c_3=0.3$).
The color map shows the simulation result for 50 samples of the random graph.
The white curves are the parameter threshold for $\rho_{SE}=0$ as in Proposition \ref{thm_sbm_thresh}, and the dashed curves represent that in the base case (b) for reference. 
In case (a), it is difficult for platform $\mathbf{A}$ to enforce strict regulation (i.e., the color map is bright and the white curve is high) when $p$ is low.} 
\label{fig:regulation_differentc_communitychain}
\end{center}
\end{figure*}

\subsection{Sympathizers in Linear Network}

For an infinite linear network, our previous arguments can be reused with minor modifications.
with $c$ replaced by $c_k$.
Therefore, a counterpart of Proposition \ref{bA_threshold_for_strictest_restriction} holds, replacing function $f_K$ in \eqref{linear_thresh_rho0} with
\vspace{-0.2cm}
\begin{equation}
    \tilde{f}_{K}(p) := \mu p^{K} - \frac{\mu c_K p^{K}}{1-p^{K+1}}.
\end{equation}

In the simulation for an infinite linear network, three users close to the sender are sympathizers with $c_i=0.21, i=0,1,2$, and the rest are non-sympathizers with $c_i=0.4, i\ge 3$.
For the purpose of comparison, simulation is conducted also for the opposite pattern, $c_i=0.4, i=0,1,2$, and $c_i=0.21, i\ge 3$.
As can be seen in Figure \ref{fig:regulation_differentc_linear} (a), it is more difficult (compared to the homogeneous case) for the platformer to impose strict regulation when the users closest to the sender are sympathizers and $p$ is low. This can be understood intuitively by noting that, when information is not diffusive, nearby users are relatively important for the sender to decide its strategy, and therefore the nearby sympathizers give it greater power.


When the information is diffusive, on the other hand, users far from the sender are relatively important, and therefore the sender needs to be more truthful to please the distant non-sympathizers (and, hence, the solid curves dips below the dashed on for large $p$ in Figure \ref{fig:regulation_differentc_linear} (a)). 

The pattern observed in panel (a) is reversed in (c), as sympathizers become helpful to the sender when $p$ is high.


\subsection{Community of Sympathizers}

For stochastic block models, we consider that users in the same community have the same payoff, but that different communities have different payoffs.
With a slight abuse of notation, 
let $c_J$ denote the payoff parameter for community $J$. 
Then, the counterpart of Proposition \ref{thm_sbm_thresh} holds as follows.
\begin{prop}
    Under Assumptions \ref{assumption_community_migration}--\ref{assumption_external_friend}, the strictest effective regulation is $\rho_{SE}=0$ if
    \vspace{-0.2cm}
    \begin{equation*}
    n_J \theta_{JJ}b_\mathbf{A}-b_\mathbf{B} 
    \ge \mu (1-c_J)p_{\phi(J)\mathbf{B}} - \frac{\sum_{L=J+1}^M R_L}{\sum_{L=1}^J R_L} \mu c_J p_{\phi(J)\mathbf{B}}
    \end{equation*}
    \vspace{-0.2cm}
    for all $J=1,\ldots,M$.
\end{prop}

Simulation is conducted for the previous model of community chain,  $n_1=n_2=n_3=30$ and \eqref{community_chain_theta}.
The basic model has $c_J=0.3$ for all $J=1,2,3$.
Keeping all other parameters fixed, a variant model has $c_1=0.21$, meaning that the sender is in a community of sympathizers.
For comparison, another variant has $c_1=0.4$, corresponding to non-sympathizers in the sender's community.
Figure \ref{fig:regulation_differentc_communitychain} indicates that in the case that the sender has sympathizers nearby, it is difficult for platform $\mathbf{A}$ to enforce strict regulation when $p$ is low. 
In the case that non-sympathizers are nearby, on the other hand, platform $\mathbf{A}$ can enforce strict regulation easily when $p$ is low.
Like linear network cases, when the information is not diffusive, the nearby users are relatively important for the sender to decide its strategy, and therefore the nearby sympathizers (non-sympathizers) bring advantage (disadvantage) to the sender.

\section{Discussion and Conclusion}




This article provided a novel perspective on the design of regulation policies that a dominant platform can enforce without reducing its user base, thus shedding light on the theoretical efficacy of such misinformation-mitigation measures, and the responsibility of mainstream platforms. 

More precisely, we identified the strictest level of regulation a platform can enforce on an information source without losing users to a non-regulated competitor. We focused on simple yet relevant network structures, such as stochastic block models, and used the model to make predictions in a variety of scenarios. In particular, we showed that a tightly connected community prevents new platform adoption thereby helping the currently dominant platform impose strict regulation, especially if the information is diffusive. We also showed that a news source in a large community tends to comply with strict regulation in the popular platform if the information is not diffusive, while a news source in a small community is more likely to relocate to an alternative platform. Finally, a news source surrounded by sympathizers who appreciate the information it emits exerts a greater indirect influence over the platform's content moderation policy than one surrounded by non-sympathizers. However, when information is diffusive, this effect diminishes, allowing the platform to enforce stricter regulations.





\section*{Appendices}

\subsection{Lemmata for Proposition 1}

\begin{lemma}
Suppose $\mathbf{P}_S=\mathbf{B}$.
In the adoption process, users can switch from platform $\mathbf{A}$ to $\mathbf{B}$ but not from $\mathbf{B}$ to $\mathbf{A}$.
\end{lemma}

\begin{proof}
We show by induction.
Suppose this statement is true at iteration $t$.
Then, for each user $i$, the number of neighbors in platform $\mathbf{A}$ (resp. $\mathbf{B}$) decreases (increases) or stays the same.
This means $\Phi_{i\mathbf{A}}^{(t)}\ge \Phi_{i\mathbf{A}}^{(t+1)}$, $\Phi_{i\mathbf{B}}^{(t)}\le \Phi_{i\mathbf{B}}^{(t+1)}$. 
Also, $p_{i\mathbf{B}}^{(t)}\le p_{i\mathbf{B}}^{(t+1)}$ because the shortest path in the subgraph for platform $\mathbf{B}$ does not become longer.
This means $\Psi_{i\mathbf{B}}^{(t)}\le \Psi_{i\mathbf{B}}^{(t+1)}$.
Note that $p_{i\mathbf{A}}^{(t)}=p_{i\mathbf{A}}^{(t+1)}=0$ and therefore $\Psi_{i\mathbf{A}}^{(t)}=\Psi_{i\mathbf{A}}^{(t+1)}$.
Hence, $V_{i\mathbf{A}}^{(t)}\ge V_{i\mathbf{A}}^{(t+1)}$ and $V_{i\mathbf{B}}^{(t)}\le V_{i\mathbf{B}}^{(t+1)}$.
Therefore, if $V_{i\mathbf{B}}^{(t)}\ge V_{i\mathbf{A}}^{(t)}$, then $V_{i\mathbf{B}}^{(t+1)}\ge V_{i\mathbf{A}}^{(t+1)}$.
This means if a user chooses $\mathbf{B}$ at iteration $t$, then chooses $\mathbf{B}$ at iteration $t+1$ too.
By induction, the statement is true at any iteration.
\end{proof}

\begin{lemma}
Suppose $\mathbf{P}_S=\mathbf{B}$ and the network is acyclic. 
User $i$ that switches to platform $\mathbf{B}$ at iteration $t$ is $t$ edges away from the sender and has $p_{i\mathbf{B}}^{(t)}=p^t$.
\end{lemma}

\begin{proof}
We show by induction.
Suppose the statement is true for $1,2,\ldots,t$.

Let $S_t$ be the set of users who switched to $\mathbf{B}$ at iteration $t$. 
Since the graph is acyclic, each user has only one neighbor upstream (closer to the sender) and the other neighbors downstream (farther from the sender).
Therefore, the upstream neighbors of $S_t$ is already in $\mathbf{B}$ at iteration $t$ (otherwise, $p_{i\mathbf{B}}^{(t)}=0$ for $i\in S_t$, which contradicts the assumption).
Therefore, the platform adoption of $S_t$ affects $N_{j\mathbf{P}}$ only if user $j$ is a downstream neighbor of $S_t$ or is already in $\mathbf{B}$.

By assumption, all users more than $t$ edges away from the sender are in $\mathbf{A}$ at iteration $t$.
Therefore, the platform adoption of $S_t$ affects $p_{j\mathbf{B}}$ only if user $j$ is a downstream neighbor of $S_t$.

Therefore, the platform adoption of $S_t$ affects $V_{j\mathbf{P}}$ only if user $j$ is a downstream neighbor of $S_t$ or is already in $\mathbf{B}$.
Therefore, if a user switches to $\mathbf{B}$ at iteration $t+1$, then it is a downstream neighbor of $S_t$.
A downstream neighbor of $S_t$ is $t+1$ edges away from the sender and has $p_{i\mathbf{B}}^{(t+1)}=p^{t+1}$.

By induction, the statement holds for any $t$.
\end{proof}

\subsection{Finite Linear Network}

We have
\vspace{-0.2cm}
\begin{equation}
\begin{aligned}
    U_\mathbf{A}(\beta) 
    = \sum_{k=0}^{n-1}(\mu +(1-\mu)\beta)p^k,\\
    U_\mathbf{B}(\beta) 
    = \sum_{k=0}^{K}(\mu +(1-\mu)\beta)p^k 
\end{aligned}
\end{equation}
with $K<n-2$ satisfying
\begin{equation}
    V_{K+1,B} < b_\mathbf{A} +(1-\mu)c \le V_{K,B}
\end{equation}
or $K=n-1$ if $b_\mathbf{A} +(1-\mu)c \le V_{n-2,B}$.
For the best strategy in platform B, the sender takes $\beta^*_\mathbf{B}=F(K)$ for $K<n-2$ and $\beta^*_\mathbf{B}=F(n-2)$ for $K=n-1$.

The strictest effective regulation $\rho_{SE}$ takes the value of zero if and only if $U_\mathbf{A}(\beta=0)\ge U_\mathbf{B}(\beta^*_\mathbf{B})$, i.e.,
\vspace{-0.2cm}
\begin{equation*}
\begin{aligned}
    &\sum_{k=0}^{n-1}\mu p^k \ge \sum_{k=0}^{K} (\mu +(1-\mu)F(K))p^k
    &\text{for }K<n-2 \\
    &\sum_{k=0}^{n-1}\mu p^k \ge \sum_{k=0}^{n-1} (\mu +(1-\mu)F(n-2))p^k
    &\text{for }K=n-1.
\end{aligned}
\end{equation*}
This leads to Proposition \ref{bA_threshold_for_strictest_restriction_finite}.

\subsection{Derivation of Proposition \ref{thm_sbm_thresh}}

Suppose the sender chooses platform $\mathbf{B}$.
The critical situation for community $J$'s platform adoption is that the first-adopting user $\phi(J)$ chooses a platform when a friend from community $1,\ldots,J-1$ is in platform $\mathbf{B}$ and other friends in community $J$ are in $\mathbf{A}$.
This user compares 
\begin{equation*}
\begin{aligned}
    V_{\phi(J)\mathbf{A}} &= n_J \theta_{JJ}b_\mathbf{A} + (1-\mu)c \quad \mbox{ and}\\
    V_{\phi(J)\mathbf{B}} &= b_\mathbf{B} + \mu p_{\phi(J)\mathbf{B}}(1-c) + (1-\mu)(1-\beta p_{\phi(J)\mathbf{B}})c.
\end{aligned}
\end{equation*}
Hence, users in community $J$ chooses $\mathbf{B}$ if $\beta\le\beta_J$ where 
\vspace{-0.2cm}
\begin{equation}\label{beta_j}
    \beta_J := \frac{1}{(1-\mu)c}\bigg(\mu(1-c) - \frac{n_J \theta_{JJ}b_\mathbf{A}-b_\mathbf{B}}{p_{\phi(J)\mathbf{B}}}\bigg).
\end{equation}
This $\beta_J$ is the sender's best strategy assuming that communities $1,\ldots,J$ are in platform $\mathbf{B}$ and $J+1, \ldots, M$ are in $\mathbf{A}$ at equilibrium. 
Then, the strictest effective regulation $\rho_{SE}$ satisfies $U_\mathbf{A}(\rho_{SE}) \ge U_\mathbf{B}(\beta_J)$, i.e.,
\vspace{-0.2cm}
\begin{equation}
    (\mu+(1-\mu)\rho_{SE}) \sum_{L=1}^M R_L 
    \ge (\mu+(1-\mu)\beta_J) \sum_{L=1}^J R_L.
\end{equation}
Plugging in $\rho_{SE}=0$ and the value of $\beta_J$, we obtain the following necessary and sufficient condition for $\rho_{SE}=0$ at the equilibrium where communities $1,\ldots, J-1$ are in Platform $\mathbf{B}$ and the rest are in $\mathbf{A}$:
\vspace{-0.2cm}
\begin{equation*}
    n_J \theta_{JJ}b_\mathbf{A}-b_\mathbf{B} 
    \ge \mu (1-c)p_{\phi(J)\mathbf{B}} - \frac{\sum_{L=J+1}^M R_L}{\sum_{L=1}^J R_L} \mu c p_{\phi(J)\mathbf{B}}.
\end{equation*}

Considering all possible assignments of communities to platforms at equilibrium, we achieve Proposition \ref{thm_sbm_thresh}.

\subsection{Numerical Validation of Assumption \ref{assumption_community_migration}}\label{appendix_validation}

To justify Assumption \ref{assumption_community_migration} for Proposition \ref{thm_sbm_thresh}, we conduct simulations, taking community chains and complete graphs of communities for example.

For variants of the base model of community chain ($n_1=n_2=n_3=30$ and \eqref{community_chain_theta}), instead of $\theta_{JJ}=3/4$, we consider $\theta_{JJ}=1/4, 1/8, 1/16$.
Figure \ref{fig:communitychain_num_userB} shows the number of users who choose platform $\mathbf{B}$ when platform $\mathbf{A}$ sets too strict regulation $\rho_\mathbf{A}<\rho_{SE}$.
For a large $\theta_{JJ}$, the number of users in platform $\mathbf{B}$ is 0, 30, 60, or 90, indicating that all users in a community choose the same platform.
On the other hand, for a small $\theta_{JJ}$, users in a community do not necessarily choose the same platform.
The simulation result is for a single sample for the random networks.

Similar results are achieved in simulation for the complete graph of communities.

\begin{figure*}
\centering
\includegraphics[width=1.00\textwidth]{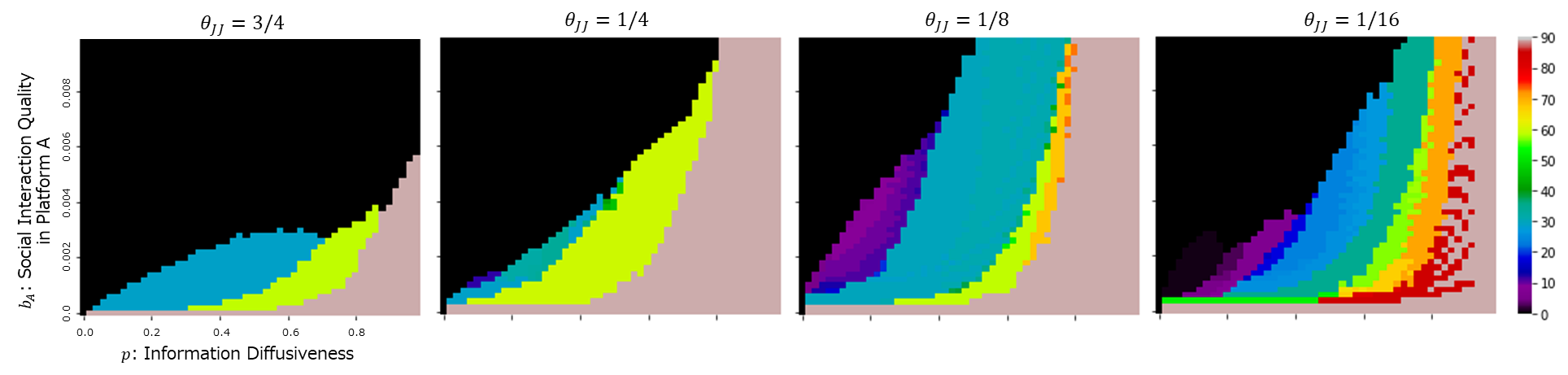}
\vspace{-.8cm}
\caption{The number of users who choose platform $\mathbf{B}$ when platform $\mathbf{A}$ imposes too strict regulation $\rho_\mathbf{A}<\rho_{SE}$.
The network is a community chain.
For a large $\theta_{JJ}$, the number of users in platform $\mathbf{B}$ is 0, 30, 60, or 90, indicating that all users in a community choose the same platform.
On the other hand, for a small $\theta_{JJ}$, users in a community do not necessarily choose the same platform.} 
\label{fig:communitychain_num_userB}

\end{figure*}

To discuss this property of community migration, we introduce new metrics.
Let $N(J,\mathbf{P})$ be the number of users in community $J$ who choose platform $\mathbf{P}$.
Then, define the number of users' irregular choices as 
\vspace{-0.2cm}
\begin{equation}
    \sum_{J=1}^M \min \{ N(J,\mathbf{A}), N(J,\mathbf{B})\}.
\end{equation}
This takes the value of zero when all users in each community choose the same platform, and the maximum value $\sum_J n_J/2$ when half users in each community choose a platform different from the other half.

Figure \ref{fig:sbm_num_userB_50samples} presents the number of users' irregular choices.
The simulation is conducted for 50 samples, fixing $p=0.7$ and $b_\mathbf{A}=0.002$.
As in Figures \ref{fig:communitychain_num_userB}, 
the results consistently indicate that when users are tightly connected within community (i.e., $\theta_{JJ}$ is high), all users in a community choose the same platform.
This simulation result supports Assumption \ref{assumption_community_migration} for Proposition \ref{thm_sbm_thresh}.

\begin{figure}
\begin{center}
\includegraphics[width=5.9cm]{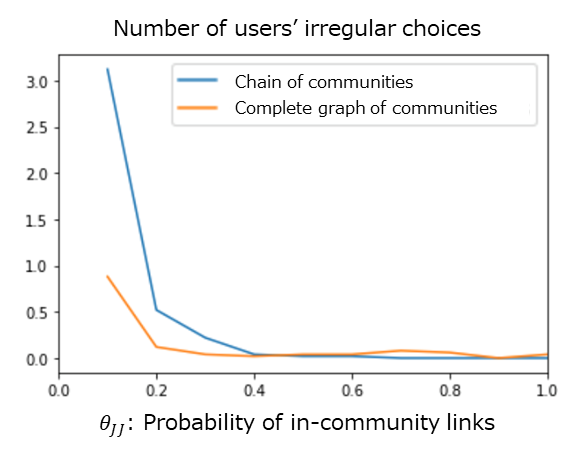}    
\vspace{-.4cm}
\caption{The number of users' irregular choices. The value of 0 means that all users in a community choose the same platform.} 
\label{fig:sbm_num_userB_50samples}
\end{center}
\end{figure}

\subsection{Community Structures Applicable to Proposition \ref{thm_sbm_thresh}}

\subsubsection{Community Chain}

The chain of communities is introduced as a physical model that satisfies \eqref{communitychain_condition}.

Heuristically, a possible way to obtain such a scaling is by considering a community chain 
with 
\begin{equation}\label{theta_jj}
    \theta_{JJ} \approx 1,
\end{equation}
\begin{equation}\label{theta_adj}
    n_J\theta_{JJ'} < 1 < n_J n_{J'} \theta_{JJ'} \ \text{  for  }\  J'=J-1, J+1,  
\end{equation}
\begin{equation}\label{theta_doubleadj}
    n_J n_{J-1} \theta_{JJ-1} n_{J+1}\theta_{JJ+1} < 1.
\end{equation}
Roughly speaking, \eqref{theta_jj} means that most user pairs in the same community are friends;
\eqref{theta_adj} means that although communities are connected, most users do not have friends in other communities;
\eqref{theta_doubleadj} means that in most cases, no user in community $J$ has friends both in community $J-1$ and community $J+1$.
Under these conditions, it should be intuitive that most users are two edges away from the nearest neighbor in an adjacent community, thus leading to \eqref{communitychain_condition}.

Assumption \ref{assumption_community_migration} is validated for \eqref{theta_jj} in simulation in Appendix D. 
The chain structure of the network guarantees Assumption \ref{assumption_migration_order}, and Assumption \ref{assumption_external_friend} is satisfied due to \eqref{theta_adj}.

\subsubsection{Star/Complete Graph of Communities}

In order to obtain \eqref{communitystar}, consider a stochastic block model with the sender in community 1 and $\theta_{1J}>0$ for all $J=2,\ldots, M$.
In addition to \eqref{theta_jj}, we assume
\begin{equation}\label{theta_1j}
    n_1\theta_{1J} < 1,\ n_J\theta_{1J} < 1,\  n_1 n_J \theta_{1J} > 1,
\end{equation}
\begin{equation}\label{theta_jk}
    \theta_{JJ'} \approx 0 \ \text{ for }\  J\neq 1, J'\neq 1, J\neq J'.
\end{equation}
Equation \eqref{theta_1j} means that in most cases, a user in community $J$ (resp. 1) does not have friends in community 1 (resp. $J$), but there are some inter-community links between 1 and $J$.
Equation \eqref{theta_jk} means that edges between peripheral communities do not exist or are negligible. 
These assumptions indicate that the first-adopting user $\phi(J)$ in community $J\ge 2$ is two-edges away from the sender and most others are three-edges away; hence \eqref{communitystar}.


Assumption \ref{assumption_community_migration} is validated for \eqref{theta_jj} in simulation in Appendix D. 
Assumption \ref{assumption_migration_order} is satisfied because $n_2\theta_{22} < n_3\theta_{33} < n_4\theta_{44}$.
The combination of \eqref{theta_1j} and \eqref{theta_jk} satisfies Assumption \ref{assumption_external_friend}.


\Urlmuskip=0mu plus 1mu
\bibliographystyle{IEEEtran}
\bibliography{brief}

\begin{IEEEbiography}
    [{\includegraphics[width=1in,height=1.25in,clip,keepaspectratio]{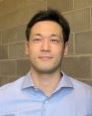}}]{So Sasaki}
is working toward the Ph.D. degree with the University of Illinois at Urbana–Champaign, Urbana, IL, USA.  He is currently with the Decision and Control Group, Coordinated Science Lab, University of Illinois at Urbana–Champaign.  His research interests include game, optimization, control, graph theory, and machine learning, with application to social media.
\end{IEEEbiography}

\begin{IEEEbiography}
    [{\includegraphics[width=1in,height=1.25in,clip,keepaspectratio]{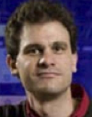}}]{Cédric Langbort}
 received the Ph.D. degree in Theoretical and Applied Mechanics from Cornell University, Ithaca, NY, USA, in 2005. He is currently a Professor of Aerospace Engineering with the University of Illinois at Urbana–Champaign, Champaign, IL, USA. He is also with the Decision and Control Group, Coordinated Science Lab (CSL), where he works on the applications of control, game, and optimization theory to a variety of fields.
\end{IEEEbiography}

\end{document}